\documentclass[12pt,oneside,reqno]{amsart}
\usepackage{amssymb}
\usepackage{array}
\usepackage[english]{babel}
\usepackage{caption}
  \usepackage[labelfont=md]{subcaption}
\usepackage{enumitem}
\usepackage[a4paper,tmargin=1in,bmargin=1.2in ,lmargin=1in,rmargin=1in,headheight=0in,headsep=0in,footskip=1.5\baselineskip]{geometry}
\usepackage{natbib,nicefrac,setspace,tabularx,graphicx,tikz,float,mathtools,bbm,paralist}
\DeclarePairedDelimiter\abs{\lvert}{\rvert}

\usepackage{verbatim} 

\usepackage[colorlinks=true,linkcolor=blue,citecolor=blue]{hyperref}

\let\oldFootnote\footnote
\newcommand\nextToken\relax

\renewcommand\footnote[1]{%
    \oldFootnote{#1}\futurelet\nextToken\isFootnote}

\newcommand\isFootnote{%
    \ifx\footnote\nextToken\textsuperscript{,}\fi}

\newtheorem{theorem}{Theorem}
\newtheorem*{theorem*}{Theorem}
\newtheorem{lemma}{Lemma}
\newtheorem{proposition}{Proposition}
\newtheorem{claim}[theorem]{Claim}
\newtheorem{corollary}[theorem]{Corollary}

\theoremstyle{definition}

\newtheorem*{definition*}{Definition}
\newtheorem*{lemma*}{Lemma}

\newcommand{\R}{\mathbb{R}}

\newcommand{\fg}{\mathfrak{g}}
\newcommand{\ff}{\mathfrak{f}}
\newcommand{\g}{\mathfrak{g}}
\newcommand{\bb}{\mathfrak{b}}
\newcommand{\eps}{\varepsilon}

\newcommand{\E}{\mathbb{E}}
\newcommand{\Prob}{\mathbb{P}}
\DeclareMathOperator*{\argmax}{arg\,max}

\def\ee{\mathrm{e}}

\renewcommand{\AA}{A}
\newcommand{\cS}{\Omega}
\newcommand{\sig}{S}
\renewcommand{\ss}{s}

\makeatletter
\renewenvironment{proof}[1][\proofname]{%
  \par\pushQED{\qed}\normalfont%
  \topsep6\p@\@plus6\p@\relax
  \trivlist\item[\hskip\labelsep\bfseries#1\@addpunct{.}]%
  \ignorespaces
}{%
  \popQED\endtrivlist\@endpefalse
}
\makeatother

\title{Learning in Repeated Interactions on Networks}
\author[]{Wanying Huang$^{\dagger}$}
\author[]{Philipp Strack$^{\ddagger}$}
\author[]{Omer Tamuz$^{\S}$}
\thanks{$^{\dagger}$California Institute of Technology. Email: \texttt{whhuang@caltech.edu}}
\thanks{$^{\ddagger}$Yale University. Email: \texttt{philipp.strack@gmail.com}}
\thanks{$^{\S}$California Institute of Technology. Email: \texttt{omertamuz@gmail.com}. Omer Tamuz was supported by a Sloan fellowship, a BSF award (\#2018397) and a National Science Foundation CAREER award (DMS-1944153).\\We would like to thank Krishna Dasaratha, Ben Golub, Hamid Sabourian and Leeat Yariv for insightful comments.}
\date{\today}

\begin{document}
\maketitle

\begin{abstract}
We study how long-lived, rational agents learn in a social network.
In every period, after observing the past actions of his neighbors, each agent receives a private signal, and chooses an action whose payoff depends only on the state. Since equilibrium actions depend on higher order beliefs, it is difficult to characterize behavior. Nevertheless, we show that regardless of the size and shape of the network, the utility function, and the patience of the agents, the speed of learning in any equilibrium is bounded from above by a constant that only depends on the private signal distribution. 
\end{abstract}

\section{Introduction}


We study social learning by long-lived agents  who observe each others' actions on a social network.
We show that information aggregation fails: the speed of learning stays bounded even in large networks, where efficient aggregation of all private information would lead to arbitrarily fast learning.
Methodologically, we introduce new techniques that allow us to relax commonly made assumptions and study general networks, with forward-looking, Bayesian agents who interact repeatedly. Repeated interactions can, for example, describe the exchange of opinions and information among friends on social media or firms that learn from each others actions. 
Arguably, social learning driven by such interactions can be an important determinant in many choice domains, such as  investments, health insurance, schools, technology adoption, or where to live and work.

More formally, we consider a group of agents who repeatedly interact with their neighbors on a network. There is a fixed but unknown state of the world, taking values from a finite set. In every period, each agent receives a private signal about the state and observes all past actions of his neighbors. Based on this information, each agent updates his beliefs, chooses one of finitely many possible actions, and receives a flow payoff that depends on his action and the state (but not on the actions of others). We consider both myopic agents who maximize their instantaneous payoff, as well as strategic agents who are forward-looking, and exponentially discount the future. Since strategic agents care about their future utilities, they may sacrifice their present flow utilities and choose actions that  induce others to behave in a way which reveals more information in the future.


The constant influx of private information in our model allows every agent to eventually learn the truth and choose the optimal action. Thus, we focus on how fast agents learn.  If the number of agents doubles, the number of private signals available to society also doubles. Hence, if information were aggregated efficiently, the time that it would take for agents to choose correctly (say, with a given high probability) would decrease by a factor of two. In other words, when information aggregation is efficient, the speed of learning increases linearly with the number of agents. The question we ask is: what is the speed of learning in equilibrium, and how does it depend on the number of agents, the structure of the network, the agents' utilities, the signals, and the discount factor? 

We focus on \emph{strongly connected} networks, where there is an observational path between every pair of agents, since otherwise efficient aggregation of information is excluded by the lack of informational channels. This is a mild assumption, and follows, for example, from the ``six degrees of separation rule'', that stipulates that there is a path of length at most six between every two members of a social network.\footnote{The science fiction writer Frigyes Karinthy proposed this rule in his 1929 short story ``L\'ancszemek'' (in English, ``Chains''). The rule was confirmed empirically on a number of online social network data sets \citep{watts1998collective,leskovec2008planetary}.}  Our main result (Theorem \ref{main theorem}) shows that the speed of learning does not increase linearly with the number of agents, and is, in fact, bounded from above by a constant that only depends on the private signal distribution, and is independent of the structure of the network and the remaining parameters of the model.


For example, consider agents who, in each period, observe an independent binary signal that is equal to a binary state with probability 0.9. Then, regardless of the number of agents and network structure, the speed of learning never exceeds ten times the speed at which an agent learns on his own. This is despite the fact that if $n$ agents shared their signals publicly, they would learn $n$ times as fast. Thus, a society of 1,000 agents who observe their neighbors' past actions does not learn faster than a society of ten agents in which information is efficiently aggregated. This means that in a society of 1,000 agents, at least 99\% of the information generated by the private signals is lost for any structure of the observational network (and any equilibrium played by the agents with any utility).

As another illustration of this result, consider the finite two dimensional grid graph: The set of agents is $\{1,\ldots,n\}^2$, and $i$ observes $j$ if and only if $|i-j|=1$. Regardless of the size of the graph, in the early periods each agent is exposed to little information, since (at most) four other agents are observed each period. Theorem~\ref{main theorem} shows that even later in the process the size of the graph does not matter substantially, as the speed of learning is bounded. 

The mechanism behind this bounded speed of learning is as follows: First, in a strongly connected network all agents learn at the same rate, as each agent could guarantee himself the same learning rate as any of his neighbors. We establish the bound on the learning rate by contradiction. Suppose that agents learn at a rate that is higher than the rate that their individual private signals can support. This implies that social information, which consists of agents’ actions, will become much more precise than each agent’s private information. As a result, agents will ignore their private signals and only rely on the social information they observe from their neighbors. This implies that agents’ actions no longer reveal any information about the state, so social information cannot grow too precise over time. This contradicts our previous hypothesis, and we conclude that agents cannot learn too fast.


The failure of information aggregation suggested by Theorem~\ref{main theorem} is an asymptotic result that describes how fast learning is in late periods. We complement this result with a numerical calculation for the first ten periods on a simple graph: the line graph with bidirectional observations.
Assuming myopic agents, and binary state, actions and signals, we find that our asymptotic lower bound on mistake probabilities holds also in the early periods.  
Quantitatively, in each of the first 10 periods agents choose correctly  with a probability that is smaller than that of a group of five agents who share their private signals.
This holds independently of the number of agents in the network.

We contribute to the social learning literature in three aspects. First, instead of focusing on short-lived agents who act only once, we consider a more realistic model in which agents are long-lived and repeatedly interact with each other. Second, we extend  existing models of social learning from myopic agents to strategic agents who discount their future utilities at a common rate. Analyzing strategic agents is complicated since these agents may have an incentive to choose a sub-optimal action today to learn more information from the actions of others in the future, whereas such an incentive is completely shut down for myopic agents.
Third, we generalize previous work on the speed of learning  on the complete network where all agents observe each other \citep{harel2021rational} to general social networks where agents only observe their neighbors.

\subsection*{Related Literature}
There are few papers addressing  repeated interaction between rational agents and its role in information aggregation. 
This is because it is challenging to analyze the evolution of beliefs of long-lived agents, particularly when these beliefs are influenced by the interactions between them. Most of the literature has focused on either short-lived agents who only act once \citep{banerjee1992simple, BichHirshWelch:92, smith2000pathological, acemoglu2011bayesian} or non-fully-rational agents  \citep{bala1998learning, molavi2018theory} and heuristic learning rules \citep{degroot1974reaching, golub2010naive, dasaratha2020learning}. Nevertheless, as many real-life situations involve repeated interactions, it is natural to study models that allow these interactions. It is likewise interesting to understand rational (and potentially forward-looking) agents as this is an important benchmark case to assess whether or to what extent failures of information aggregation are driven by a lack of patience or lack of rationality of the agents.

Among models which consider repeated interactions, the bandit literature studies endogenous information acquisition and the resulting free rider problem when multiple agents try to learn the state from each other's public signals \citep{bolton1999strategic,keller2005strategic,keller2010strategic,heidhues2015strategic}. \cite{bala1998learning} extend this to a social network setting for non-fully-rational agents who do aggregate the results of their neighbors' experiments but disregard the information contained in their choices. They examine how the geometry of the social network affects learning outcomes.

Focusing on a repeated interaction setting in which agents and have no experimentation motives, \cite{mossel2014asymptotic} consider rational but myopic agents who observe each other's actions. They give conditions for learning to occur on infinite undirected graphs. \cite{mossel2015strategic} further generalize their setting to allow for forward-looking agents.\footnote{In a similar setting, \cite{migrow2022strategic} shows that for two forward looking agents who observe each other there does not exist an equilibrium where agents behave myopically, under some assumptions on the signal structure.}
Unlike our model, agents in these models only receive one signal at the beginning of time.
They do not study the speed of learning, and instead focus on identifying the types of social networks in which learning always occurs.

Complementing the previous literature, we take the next natural step: we ask how fast learning occurs and study its relationship with the size of the network. Furthermore, we consider agents with any discount factor with myopic agents as a special case. To our best knowledge, this is the first paper to consider social learning in a network setting with fully-rational  agents who interact repeatedly.  

A recent paper that considers rational agents in a repeated setting is by \cite*{harel2021rational}, who study the speed of learning when all agents are myopic and observe each other. They show that similar to our main result, for any number of agents, the speed of learning from actions is bounded above by a constant. 
Their proofs rely crucially on the symmetry inherent in the complete network in which all agents observe each other and actions are common knowledge. Their analysis  relies on a phenomenon called ``groupthink'': a feedback loop that develops when  all players take the wrong action,  observe that everyone else also took the same action, become more confident in their wrong beliefs, and then again take the wrong action. In incomplete networks, actions are no longer common knowledge, making the techniques used in the aforementioned paper inapplicable in our setting.   The techniques we introduce furthermore allow us to consider non-myopic agents and  multiple states of the world.

Our main insight is that learning cannot be too fast since fast individual learning would cause agents to ignore their private signals. This is reminiscent of the information cascades that drive the failure of information aggregation in the classical herding literature \citep{BichHirshWelch:92, banerjee1992simple, smith2000pathological}.  However, the force in our model affects the speed of learning rather than determining whether or not information aggregates, and because of the repeated interactions between agents, it requires a very different analysis. A similar insight also appears in the earlier literature on rational expectations in financial markets, where it implies the breakdown of the efficient market hypothesis \citep{grossman1980impossibility}: prices cannot fully reflect all available information, precisely because if they did, it would eliminate the agents' incentive to respond to their private information, contradicting the assumption that prices contain all information. This is known as the Grossman-Stiglitz paradox.

Following the herding literature, our paper studies the friction that arises for information aggregation when  actions are coarse and thus do not fully reveal beliefs.
Another strand of the literature shows that information aggregation may still be slow, even with a continuous action space. For example, \cite{vives1993fast} considers a Cournot competition model with a common unknown production cost among a continuum of firms. He 
shows that noisy observations of past actions (through market prices) slow down the speed of learning. Although his environment is different from ours,\footnote{More specifically, his setting differs from ours along a number of dimensions: Firms learn from public prices, which are noisy observations of the average actions of others. Furthermore, there is no network and since there is a continuum of agents, any strategic incentive is also abstracted away.} the force behind his slow speed of learning is related to ours: By examining the asymptotic behavior of public information's informativeness (reflected in prices) and agents' responsiveness to their private information, he found that as public information becomes more informative, agents respond less to their private information, leading to a slower increase in how informative the public information can be. Put in his words, ``... information revelation through the price system can be slow precisely because it is successful.'' We contribute to this literature by showing that this intuition in these specific cases generalizes extensively, especially in a general network setting with forward-looking agents.

With a rich-enough action space that fully reveals agents' beliefs, \cite{dasaratha2019aggregative} study how different generational networks affect the learning rate with Gaussian signals and continuous actions. They find that learning is slow for large symmetric overlapping generations networks with a uniform bound on the number of signals aggregated per generation. The main mechanism behind the inefficiency is a confounding effect: 
As earlier generations observe common predecessors, their actions are correlated, reducing the amount of information transmitted to the next generation. This effect is inherent in their overlapping generations network structure, in which information travels unidirectionally. Since information travels bidirectionally in our model, higher-order beliefs pose an obstacle that is not present in \cite{vives1993fast} and \cite{dasaratha2019aggregative}.





\section{Model}
Let $N=\{1, 2, \ldots, n\}$ be a finite set of agents. Time is discrete and the horizon is infinite, i.e.\ $t \in \{1, 2, \ldots\}$. In every time period $t$, each agent $i$ has to choose an action $a_t^i$ from a finite set $\AA$. There is an unknown state of the world $\omega$, taking values in a finite set $\cS$, that is chosen at period 0 and does not change over time. Each state occurs with strictly positive probability. We denote by $\fg,\ff$ generic elements of $\cS$. 

\subsection{Utility and Optimal Actions}
The flow utility for choosing an action $a \in \AA$ when the state is $\fg \in \cS$ is $u(a,\fg)$ for some $u \colon \AA \times \cS \to \R$.
We assume that for each state $\fg \in \cS$ there is a unique action  $a^\fg \in \AA$,  which maximizes $u(\cdot,\fg)$, the utility in that state:
\begin{align*}
    \{a^\fg\} = \argmax_{a \in \AA}u(a,\fg).
\end{align*}
We also assume that these optimal actions are distinct, i.e.\ if $\fg \neq \ff$ then $a^\fg \neq a^\ff$.
These assumptions facilitate learning from actions in the sense that observing the optimal action of an agent who knows the state allows other agents to infer the state.
Note that these assumptions hold for any generic utility $u$, as long as the set $\AA$ has at least as many elements as $\cS$. When there are two states then this assumption is necessary to make the model non-trivial, since otherwise there is a dominant action which will always be played in every equilibrium.\footnote{For three states or more this assumption is more restrictive. Indeed, our Lemma~\ref{lemma:strategic-myopic}, which shows that agents eventually choose myopic actions, may not hold without it. This is because even when beliefs concentrate around a state, multiple actions can be optimal for particular likelihood ratios about  the other states, and so non-myopic behavior can persist.}


An important example which the reader may wish to keep in mind is the case of binary states and actions, and where the agent aims to match the state: $\cS=\{\fg,\ff\}$, $\AA = \{a^\fg,a^\ff\}$ with $u(a^\fg,\fg)=u(a^\ff,\ff)=1$ and $u(a^\fg,\ff)=u(a^\ff,\fg)=0$. This setting already features all the forces and tensions of the general case, and likewise offers the same conclusions. 


\subsection{Agents' Information} In each period $t$, agent $i$  receives a private signal $s_t^i$ drawn from a finite set $\sig^i_{t}$. Conditional on the state $\omega$, signals $s_t^i$ are independent across agents and time, with distribution $\mu^{i,\omega}_{t} \in \Delta(\sig^i_{t})$. For distinct $\ff,\fg \in \cS$, the distributions $\mu^{i,\fg}_{t}$ and $\mu^{i,\ff}_{t}$ are distinct and  mutually absolutely continuous, so that no signal excludes any state or perfectly reveals the state. Thus, the log-likelihood ratio of any signal $\ss$ 
\[
\ell^{i,\fg,\ff}_{t}(\ss) = \log \frac{\mu^{i,\fg}_{t}(\ss)}{\mu^{i,\ff}_{t}(\ss)}
\]
is well defined. 

We focus on \emph{bounded signals} in the sense that the private belief induced by any signal cannot be arbitrarily strong. More specifically, we assume that there exists a constant $M > 0$ that bounds the absolute value of the likelihood induced by any signal:
\begin{equation}\label{eq:def-M}
 M = 2\sup_{\ff,\fg,i,t,\ss}\abs*{ \ell^{i,\fg,\ff}_{t}(\ss)} \,.
\end{equation}
We make this assumption for tractability and discuss its relaxation in the conclusion. We allow signals to depend on calendar time and the agents' identities to highlight the robustness of our results. However, to understand our main economic insight, it suffices to think of the setting where all signals are i.i.d.\ across agents and time, as is typically assumed in the literature.

For each agent $i$ there is a subset of agents $N_i \subseteq N$ who are his {\em social network neighbors}, and whose actions he observes. We include $i\in N_i$ since agent $i$ observes his own actions. The information available to agent $i$ at time $t$, before taking his action $a_t^i$, thus consists of a sequence of private signals $(s_1^i, \cdots, s_t^i)$ and the history of actions observed by $i$,
\[
    H_t^i = \{a_s^j: s < t, j \in N_i\} \,.
\] Let $\mathcal{I}_t^i = \sig_1^i \times \cdots \times \sig_t^i \times \AA^{|N_i|\times (t-1)}$ so that the private history $I^i_t= (s_1^i, \cdots, s_t^i, H_t^i)$ is an element of $\mathcal{I}^i_t$. 

We assume  that the social network is {\em strongly connected}: there is an observational path between every pair of agents (we relax this assumption in \S\ref{sec:not-strongly-connected}).\footnote{Formally, for each $i,j \in N$ there is a sequence $i=i_1,i_2,\ldots,i_k=j$ such that $i_2 \in N_{i_1}, i_3 \in N_{i_2},\ldots i_{k} \in N_{i_{k-1}}$.}$ ^{,}$\footnote{In a directed Erdős–Rényi graph with $n$ agents, when the expected number of neighbors is $\ln n+c$ for large $c$, then with high probability the graph will be strongly connected, and so our assumption will be satisfied with high probability \citep[see Theorem 5 in][]{graham2008note}.} This assumption avoids situations in which efficient aggregation of information is precluded because there is no channel for information to travel from $i$ to $j$. 
Furthermore, we assume that agents observe their neighbors in every period. This assumption is purely to simplify notation; we discuss the case where an agent observes their neighbors only at intermittent and potentially random times in Appendix~\ref{app:random-time}.

\subsection{Strategies and Payoffs}
A pure strategy of agent $i$ at time $t$ is a function $\sigma_t^i: \mathcal{I}^i_t \to \AA$. A pure strategy of agent $i$ is a sequence of functions $\sigma^i=(\sigma^i_1, \sigma^i_2, \cdots)$ and a pure strategy profile is a collection of pure strategies of all agents, $\sigma= (\sigma^i)_{i \in N}$. We write $\sigma = (\sigma^i, \sigma^{-i})$ for any agent $i \in N$, where $\sigma^{-i}$ denotes the pure strategies of all agents other than $i$. Given a pure strategy profile $\sigma$, the action of agent $i$ at time $t$ is $a^i_t(\sigma) = \sigma^i_t(I_t^i).$
The flow utility of agent $i$ at time $t$ is 
\[
    u(a^i_t(\sigma),\omega).
\]
Agents do not observe their flow utilities. Nevertheless, one can incorporate  observations of utilities into the private signals that agents receive. A special case of our model is one in which the agent receives an observed payoff $v(a_t^i,s_t^i)$ that depends on the action and the signal realization. In this case, the flow utility corresponds to the expected observed payoff $u(a,\omega) = \sum_{s} \mu^{\omega}(s) \cdot v(a,s)$. The agent always observes their payoff as it is only a function of their signal; yet, from an ex-ante perspective, the situation is identical to that of our setting.
The assumption of unobserved flow utilities is commonly made in the literature to model learning without an experimentation motive. Indeed, any learning situation without an experimentation motive can be reduced to such a situation \citep[see, e.g., the discussion in][\S 2.1, p.\ 981]{rosenberg2009informational}.

We assume that agents discount their future utilities at a common rate $\delta \in [0, 1)$. The expected utility of agent $i$ under strategy profile $\sigma$ is thus
\begin{align*}
u^i(\sigma) & = (1-\delta)\sum_{t=1}^{\infty} \delta^{t-1}
\E[u(a^i_t(\sigma),\omega)].
\end{align*}
We call agents with $\delta=0$ \emph{myopic} and agents with $\delta>0$ \emph{strategic}. Myopic agents fully discount future payoffs and choose their actions in each period to maximize expected flow utilities.

Regardless of whether agents are strategic or myopic, the sole benefit of observing others' past actions is to learn about the state. That is, others' actions reflect the signals they receive, and thus observing others' actions can help agents make better inferences about the state. This pure informational motive is an important feature of the model: each agent's flow utility depends only on his own actions and the state, and it is independent of the actions of the others. 

\subsection{Equilibrium}
This is a game of incomplete information, in which agents may have different information regarding the underlying unknown state and the actions of others.  We use Nash equilibrium as our equilibrium concept and refer to it as equilibrium thereafter.\footnote{Our results, which apply to all Nash equilibria, thus also apply to any refinement of Nash equilibrium such as sequential equilibrium.} The existence of a (mixed) equilibrium is guaranteed in this game by standard arguments, since, in the product topology on strategies, the space of strategies is compact and utilities are continuous. We note here that every mixed equilibrium can be mapped to a behaviorally equivalent pure equilibrium by adding to each agent's private signal an additional component that is independent of the state and all other signals, and assuming that the agent uses this signal to randomly choose between actions.\footnote{
Formally, for any signal space $\sig_t^i$ we can consider $\tilde{\sig}_{t}^i = \sig_t^i \times \AA^{|\mathcal{I}_t^i|}$ with the signal distributions $\tilde{\mu}^{i, \fg}_{t}$ equal to the product measure of $\mu^{i, \fg}_{t}$ and  $|\mathcal{I}_t^i|$ independent random variables each taking each value $a \in \AA$ with probability $\Prob[a_t^i = a |I_t^i]$. As this transformation does not affect the informational content of the signals it leaves the constant $M$ unchanged. Furthermore, we can replicate any behaviorally mixed strategy by the pure strategy that takes the action $a$ if and only if the entry of the second component corresponding to the private history $I_t^i$ equals $a$.}
As our results will only depend on the information about the state contained in the signal, it thus suffices to establish them for pure strategy equilibria, to show that they hold for all (pure and mixed) equilibria.

As usual, a pure strategy profile $\sigma$ is an equilibrium if no agent can obtain a strictly higher expected utility by unilaterally deviating from $\sigma$. That is, a pure strategy profile $\sigma$ is an \emph{equilibrium} if for all agents $i$, and all strategies $\tau^i$
\[
u^i(\sigma^i, \sigma^{-i}) \geq u^i(\tau^i, \sigma^{-i})\,.
\]

\subsection{Speed of Learning} We say that agent $i$ chooses correctly at time $t$ if $a^i_t = a^\omega$, i.e., if the agent chose the action that is optimal given the state.
We measure the speed of learning of agent $i$ by the asymptotic rate  at which he converges to the correct action \citep[see, e.g.,][]{vives1993fast, molavi2018theory, hann2018speed,
rosenberg2019efficiency, harel2021rational}. Formally, the {\em speed of learning} of agent $i$ is
\begin{align} \label{def:speed}
   \liminf_{t \to \infty} -\frac{1}{t}\log \Prob[a_t^i \neq a^\omega].
\end{align}
If this limit exists and is equal to $r$, then the probability of mistake at large times $t$ is approximately $\ee^{-rt}$. As we explain below in \S \ref{sec:benchmarks}, this is the case for the benchmark case of a single agent who receives conditionally i.i.d.\ private signals at each period.


\section{The Public Signals Benchmark} \label{sec:benchmarks}
As a benchmark, we briefly discuss the case of \emph{public} signals for a single or multiple agents. We also assume that signals are i.i.d.\ across time and agents, with $\mu^{\fg}=\mu^{i, \fg}_{t}$ for  all $\fg \in \cS$. In the single-agent case, a classical large deviations argument shows that the limit $r_a = \lim_{t \to \infty} -\frac{1}{t}\log\Prob[a^i_t \neq a^\omega]$ exists and is positive.\footnote{For textbook treatments see, e.g., pp.\ 380-384 in \cite{cover2006elements} for the binary state case or, for the general finite state case, Theorem 2.2.30 in \cite{dembo2009large}. See \cite{moscarini2002law} for an application in economics to single agent decision problems, and \cite{frick2022learning} for an application to a multi-agent setting.}  Note that the fact that the limit is positive implies that the agents learn the state; the probability of choosing incorrectly tends to zero. This is a consequence of the assumption that the measures $\mu^\fg$ are distinct which ensures that signals are informative.

Next, consider the case where each of $n$ agents observes all $n$ independent public signals in each period, as well as their neighbors' past actions. As actions contain no additional information about the state relative to the signals, this situation is identical to the single-agent case, except that now each agent receives $n$ independent signals at each period. An agent in period $t$  will thus have observed exactly as many signals as a single agent in autarky in period $n \cdot t$.
It thus follows from the single-agent case that when signals are public, the speed of learning for $n$ agents is $n\cdot r_{a}$.

These results for the case of public signals immediately bound the speed of learning in the private signals case:
In any social network, observed actions contain weakly less information than the private signals. Thus, $n\cdot r_{a}$ is an upper bound to the speed of learning for any network with $n$ agents and private signals. 

\section{Results}
\label{sec:results}

We now state our main result. It turns out that in a strongly connected network, all agents learn at the same speed (by Lemma \ref{lem:rate-bound-from-observing} in \S \ref{sec:analysis}), and we call this common speed of learning \emph{the equilibrium speed of learning.} In contrast to the public signal case, our main result shows that regardless of the size of the network, the equilibrium speed of learning  is bounded above by a constant.
Recall that in \eqref{eq:def-M} we defined $M/2$ to be the maximal log-likelihood ratio induced by any signal.

\begin{theorem} \label{main theorem}
The equilibrium speed of learning is at most $M$, in any equilibrium, on any social network of any size, for any discount factor $\delta \in [0, 1)$, and any utility $u$. 
\end{theorem}

Perhaps surprisingly, Theorem \ref{main theorem} shows that adding more agents (thus more information) to the network and expanding the network cannot improve the speed of learning beyond some bound, which is twice the strength of the strongest possible signal, as measured in log-likelihood ratios. Indeed, this upper bound on the learning speed implies that 
more and more information is lost as the size of the network increases. 

For example, for a binary state, binary actions and independent binary signals that are equal to the state with probability 0.9, the speed of learning in a social network of any size is bounded by that of ten agents who observe each other's signals directly.\footnote{To see this, given the signal distribution, we calculate the speed of learning in the single-agent case, which is approximately equal to $0.51$ (see \cite{harel2021rational} for exact expressions for the speed of learning in the binary state case). As discussed in \S \ref{sec:benchmarks}, the learning speed in a network of ten agents with public signals is ten times that of the single-agent case. From Theorem \ref{main theorem}, the upper bound $M$ to the equilibrium speed of learning  is approximately $4.4$, which is less than ten times $0.51$.}
Consequently, in any social network, even if there are 1,000 agents who  observe their neighbors' past actions, they cannot learn faster than a group of ten agents who share their private signals. Equivalently, their speed of learning cannot be more than  ten times that of a single agent. Thus almost all of the private information in large networks is lost, resulting in inefficient information aggregation. 

The idea behind our proof of Theorem \ref{main theorem} is as follows. Intuitively, one might think that larger networks would boost the speed of learning as agents acquire  more  and more information from their neighbors, as well as indirectly from their neighbors' neighbors etc. However, we argue that the social information gathered from observing neighbors' past actions cannot be too precise. Indeed, if that were the case, 
agents would base their decisions only on the social information. As a result, their actions would no longer reveal any information about their private signals so that information aggregation would cease. Thus, social information cannot grow to be much more precise than private information. But if agents learn quickly, then their actions provide very precise social information. Hence, we conclude that agents cannot learn too quickly.

In sum, regardless of the size of the network, private information must continue to influence agents' decision-making, which can only happen if the social information is not too precise, which in turn can only happen if agents do  not learn quickly. Moreover, as we state in the next section, $M$ is an upper bound to how fast the precision of private information increases with time (see Lemma \ref{lemma:private_signal} in \S \ref{sec:analysis}), and this bound, too, is independent of the network size. Combining these insights, we conclude that the speed of learning in a social network of any size is bounded by $M$.

\subsection{Numerical Calculation on the Line Graph}
While Theorem~\ref{main theorem} shows that information aggregation fails in the long run, it leaves open the question of what happens in early periods.
Clearly, if many agents all observe each other, much information could be aggregated already in the first period, as the first period actions can reveal many independent pieces of information.
The answer to this question becomes less clear in a setting where the number of neighbors is bounded, even if there are many agents in total.

To supplement the asymptotic result of Theorem~\ref{main theorem} we consider agents who observe both of their adjacent neighbors on a line, i.e., $N_i = \{i-1,i,i+1\} \cap N$.
We study a binary state and binary action setting with a uniform prior and assume that in each period, each agent gets a conditionally independent and identically distributed symmetric binary signal that is equal to the state with probability $q$.
We consider  myopic agents, i.e., $\delta=0$, and the tie-breaking rule under which agents follow their first signal when they are indifferent.  

Under these assumptions we calculate the exact probabilities of mistakes in the first $10$ periods.  A naive calculation would require considering some $10^{30}$ possible signal realizations, which is not feasible.\footnote{On a bidirectional line the number of signals that could (potentially indirectly) influence an agent's period $t$ action is the sum of her total number of private signals up to time $t$ and the total number of signals observed by his $t-1$ neighbors in each direction, i.e. $ t+\sum_{s \leq t-1} (2 s) = t + t (t-1) = t^2 $. When $t=10$, this would yield $2^{100} \approx 1.3 \times 10^{30}$ signal  realizations.}  To approach the computational problem, we use the ``dynamic cavity algorithm'' proposed by \cite{kanoria2013tractable} for calculating Bayesian beliefs in social learning environments on tree graphs, which exploits the fact that conditioning on the state and a given agent's actions makes the actions of his left-hand neighbor independent of the actions of his right-hand neighbor.
As such a decoupling argument is not available for graphs with cycles, it seems computationally infeasible to perform a similar calculation for, e.g., the two dimensional grid.

We focus on the agents who are not close to the ends of the graph: All agents $i \in \{11,12,\ldots,n-12,n-11\}$ face the same decision problem in the first ten periods, and we calculate their probabilities of choosing the wrong action. Note that these probabilities are independent of the number of agents $n$. Equivalently, these are the error probabilities of any agent on a bi-infinite line graph.

\begin{figure}[t!]
\centering
\includegraphics[width=\textwidth]{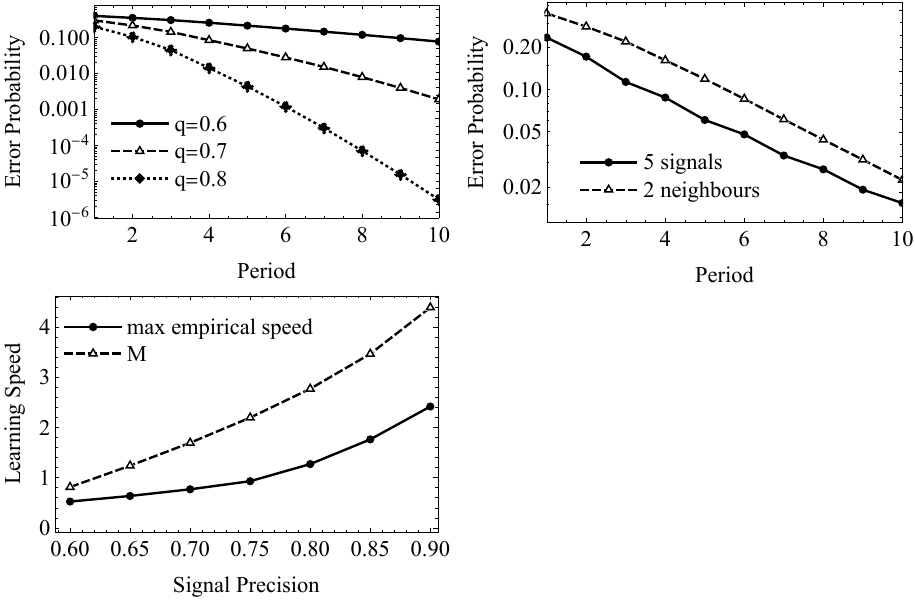}
\caption{On the top left: The log of the inverse of probability of error $-\log \mathbb{P}[a^i_t \neq a^\omega]$ as a function of the period  for agents on the bidirectional line graph (away from its ends) for different probabilities $q$ of the signal matching the state. The top right picture compares the error probabilities when the agent observes their 2 neighbors' actions and 4 other agents' signals for $q=0.65$.
On the bottom: The maximal empirical learning speed $\max_{t \in \{2,\ldots, 10\}}-\log \frac{1}{t}\mathbb{P}[a^i_t \neq a^\omega]$ for different precisions of the signal on the $x$-axis.\label{fig:line-learning}}

\end{figure}

The results are depicted in Figure~\ref{fig:line-learning}. 
On the top left, we plot the evolution of error probabilities on a log-scale for different precision of the signals $q$.  Since the ordinate uses a logarithmic scale, an exponential decay of error probabilities would manifest as a downward sloping straight line. Indeed, the graph shows that this decay is approximately exponential from the very early periods. This suggests that at least for myopic agents (the limiting case of very impatient agents), our asymptotic results bounding the speed  of learning may begin to apply early on, even before these agents have learned the state very precisely. 

On the bottom left we plot, for different signal precisions, the maximum of the empirical learning speed $\max_{t \in \{2,\ldots, 10\}}-\log \frac{1}{t}\mathbb{P}[a^i_t \neq a^\omega]$ starting from the second period as well as our asymptotic bound of $M$.
As one can see the asymptotic bound we obtained in Theorem~\ref{main theorem} holds for different precision of the signal starting from the second period.
Thus, in this example, the conclusion of Theorem~\ref{main theorem} do not only hold asymptotically, but already in early periods.
For comparison, we also plot the maximum of the rate of mistake for the case where five agents share their private signals on the top right, for $q=0.65$. The figure shows that in the first ten periods agents do worse on the bidirectional line than they would if they directly observed the private signals of their two neighbors on the right and two neighbors on the left.\footnote{While not depicted in the figure, this still holds for any precision of the signal $q \in \{0.6,0.7,0.8,0.9\}$.}


\section{Analysis} \label{sec:analysis}

In this section we provide a detailed analysis of the agents' beliefs and behavior,leading to a proof sketch for Theorem~\ref{main theorem}.

\subsection{Agents' Beliefs}
Let $p^{i,\fg}_t$ denote the posterior belief of agent $i$ assigned to event $\omega = \fg$ after observing $I_t^i$, i.e. $p^{i, \fg}_t = \Prob[\omega=\fg|I_t^i]$. The log-likelihood ratio of agent $i$'s posterior beliefs of the state being $\fg$ over the state being $\ff$ at time $t$ is 
\begin{equation}
L_t^{i,\fg,\ff} =\log\frac{p_t^{i,\fg}}{p_t^{i,\ff}}=\log\frac{\Prob[\omega=\fg|s^i_1,\ldots,s^i_t, H^i_t]}{\Prob[\omega=\ff|s^i_1,\ldots,s^i_t, H^i_t]}.     
\end{equation}
Then, it follows from Bayes' rule that this log-likelihood ratio of agent $i$'s posterior beliefs at time $t$ is equal to
\begin{align*}\label{eq:LLR}
  L^{i,\fg,\ff}_t 
  &= \log\frac{\Prob[\omega=\fg]}{\Prob[\omega=\ff]}
    + \log \frac{\Prob[H^i_t | \omega=\fg]}{\Prob[H^i_t | \omega=\ff]}
    + \log\frac{\Prob[s_{1}^i,\ldots,s_t^i|H^i_t,\omega=\fg]}{\Prob[s_{1}^i, \cdots, s_{t}^i|H^i_t,\omega =\ff]}.
\end{align*}
We call  
\begin{align*}
 Q^{i,\fg,\ff}_t = \log\frac{\Prob[\omega=\fg]}{\Prob[\omega=\ff]}
    +\log \frac{\Prob[H^i_t | \omega=\fg]}{\Prob[H^i_t | \omega=\ff]},
 \end{align*}
the {\em social likelihood} of agent $i$ at time $t$. This is 
the log-likelihood ratio of the social information observed by agent $i$.
Intuitively, $Q_t^i$ measures the inference an outside observer would draw from the observations of the actions of $i$ and his neighbors, without observing $i$'s private signals.
Similarly, we call
 \begin{align*} P^{i,\fg,\ff}_t = \log\frac{\Prob[s_{1}^i,\ldots,s_t^i|H^i_t,\omega=\fg]}{\Prob[s_{1}^i, \cdots, s_{t}^i|H^i_t,\omega =\ff]},
\end{align*}
the {\em private likelihood}  agent $i$ at time $t$. Thus, we can write the log-likelihood ratio of agent $i$'s posterior beliefs at time $t$ as 
\begin{equation} \label{eq:beliefs}
L^{i,\fg,\ff}_t = Q^{i,\fg,\ff}_t + P^{i,\fg,\ff}_t,   
\end{equation}
which is the sum of his social likelihood and his private likelihood. We call $L^{i,\fg,\ff}_t$ the \emph{posterior likelihood} of agent $i$ at time $t$.



\subsection{Agents' Behavior}
 
In the context of a strategy profile $\sigma$, the \emph{myopic action} of agent $i$ at time $t$ is
\begin{align*}
    m^i_t \in \argmax_{a \in \AA} \E[u(a,\omega)|I^i_t].
\end{align*}
This is the action that maximizes the expected flow utility given the information available at that time, and hence it is the action that a myopic agent would take.\footnote{We assume some deterministic tie-breaking rule when the agent is indifferent. Our results do not depend on this choice and would follow for any tie breaking rule that is common knowledge.} 

In contrast to a myopic agent, a strategic agent may not always choose the myopic action in equilibrium. Indeed, since a strategic agent is forward-looking, in each period he faces a trade-off between choosing the myopic action and strategically experimenting by choosing a non-myopic action. On the one hand, 
he needs to bear the immediate cost associated with a non-myopic action. On the other hand, choosing a non-myopic action may allow him to elicit more information from his neighbors' future actions, which he could then use to make better choices in the future. Hence, when the informational gain from experimenting exceeds the current loss caused by a non-myopic action, a strategic agent has an incentive to experiment.

Nevertheless, we show that when an agent's belief of any state is close enough to zero or one, he chooses the unique optimal action at that state which is also the myopic action in equilibrium. This holds despite the fact that the agent is forward-looking. Intuitively, if a strategic agent is very confident about the state, he is expected to pay a high cost if he chooses a non-myopic action and experiments. Consequently, as his expected future gain will not exceed the expected current loss from experimenting, he has no incentive to experiment. 

\begin{lemma}[Myopic and Strategic Behavior]\label{lemma:strategic-myopic}\hfill
\begin{compactenum}[(i)]
\item There is a constant $c>0$, independent of $\delta$ such that, in any equilibrium, if it holds for $\fg \in \cS$ that $p^{i,\fg}_t > \frac{c}{c+1-\delta}$, then $a^i_t = m^i_t = a^\fg$.
\item There exists a random time $T<\infty$ such that in equilibrium, all agents behave myopically after $T$, i.e. $t \geq T \Rightarrow a^i_t = m^i_t$ for all $i$ almost surely.
\end{compactenum}
\end{lemma}

The first part of Lemma \ref{lemma:strategic-myopic} applies to a fixed discount factor $\delta$, rather than asymptotically to $\delta$ tending to one. So, agents need not have learned the state very precisely at the point in which they become myopic. However, it does imply that as agents become more patient, i.e., $\delta$ increases, the posterior belief threshold $\frac{c}{c+1-\delta}$ for choosing the myopic action becomes closer to certainty. Indeed, as $\delta$ approaches $1$, agents value their future utilities more and the incentive of experimenting becomes stronger.  For these agents to forgo their potential expected future informational gains and choose the myopic action, they must be fairly confident about the state. The second part of Lemma \ref{lemma:strategic-myopic} states that in finite time the belief of all agents will be sufficiently precise such that they all behave myopically in all future periods.

The next lemma shows that  in equilibrium, each agent learns weakly faster than any of his neighbors. Denote the equilibrium speed of learning of agent $i$ by $r_i$.
\begin{lemma}[All agents learn at the same speed]\label{lem:rate-bound-from-observing}\hfill
\begin{compactenum}[(i)]
\item If agent $i$ can observe agent $j$, i.e. $j \in N_i$, then  in  equilibrium $i$ learns weakly faster than $j$, i.e. $r_i \geq r_j$.

\item All agents learn at the same speed, i.e. $r_i = r_j$ for all $i,j$, in any strongly connected network.
\end{compactenum}
\end{lemma}

We will henceforth call the common speed of learning in a strongly connected network \emph{the equilibrium speed of learning}. The proof of this lemma relies on an extension of the \emph{imitation principle} from myopic agents to strategic agents. For myopic agents, the imitation principle states that if $i$ observes $j$, then $i$'s actions are not worse than $j$'s:
\begin{align*}
    \E[u(a_t^i,\omega)] \geq \E[u(a_{t-1}^j,\omega)],
\end{align*} 
since $i$ can always imitate $j$ \citep[similar arguments are used in][]{sorensen1996rational,smith2000pathological,gale2003bayesian, golub2017learning}.
We show that for strategic agents, in equilibrium, $i$'s actions are never much worse than $j$'s, even though $i$ may choose myopically sub-optimal actions. To formalize this we denote by $\bar{u}=\E[u(a^\omega,\omega)]$ the expected utility of a decision maker that knows the state. For any $i$ and $t$ it holds that $\E[u(a^i_t,\omega)] \leq \bar{u}$. We can think of $\bar{u}-\E[u(a^i_t,\omega)]$ as the loss in flow utility as compared to the first-best. The imitation principle for strategic agents takes the following form:
\begin{align*}
 \bar{u}-\E[u(a_t^i,\omega)] \leq \frac{1}{1-\delta}(\bar{u}-\E[u(a_{t-1}^j,\omega)]).   
\end{align*}
It is obtained by upper-bounding the agent's loss by the loss he would obtain if guessing correctly in every future period and observing that this loss must be less than the loss obtained by taking the action agent $j$ took last period in every future period.
By Claim~\ref{lem:W} in the Appendix, this upper bound on the loss implies that there is some constant $c > 0$ such that the probability of choosing incorrectly is bounded:
\begin{align*}
 \Prob[a_t^i \neq a^\omega] \leq \frac{c}{1-\delta}\Prob[a_{t-1}^j \neq a^\omega].   
\end{align*}
 One can easily see that when $\delta =0$, the above equation coincides with the imitation principle for myopic agents.



\subsection{Social and Private Beliefs}
We analyze the agents' beliefs by decomposing their likelihoods into the private and social parts. Recall that by \eqref{eq:beliefs}, the posterior likelihood of agent $i$'s at time $t$, $L^{i, \fg, \mathfrak{f}}_t$, is equal to $Q^{i, \fg, \mathfrak{f}}_t + P^{i, \fg, \mathfrak{f}}_t$, the sum of the social and the private likelihoods. We are interested in the sign of $L^{i, \fg, \mathfrak{f}}_t$ as it determines the corresponding myopic action: $m^i_t$ equals $\mathfrak{g}$ if $L^{i, \fg, \mathfrak{f}}_t \geq 0$, and  $\mathfrak{f}$ otherwise. Let us first focus on the first component of $L^{i, \fg, \mathfrak{f}}_t$: agent $i$'s social likelihood $Q^{i, \fg, \mathfrak{f}}_t$. In the following lemma, we establish a relationship between the equilibrium speed of learning and the precision of the social likelihood, which is crucial in proving our main theorem.


\begin{lemma}  \label{prop1}
Suppose that the equilibrium speed of learning is at least $r$.  Then, conditioned on  $\omega=\fg$, it holds for any $\ff \neq \fg$
\begin{align*}
    \liminf_{t \to \infty}\frac{1}{t}Q^{i,\fg,\ff}_t \geq r \quad \text{almost surely.}
\end{align*}
\end{lemma}

This lemma states that a high learning speed implies that the social information inferred from a given agent $i$ and his neighbors' actions must become precise at a high speed. Intuitively, if agents learn quickly, then their actions provide very precise social information. 

The proof of Lemma~\ref{prop1} uses the idea of a fictitious outside observer who observes the same social information as agent $i$ and nothing else. Since he observes $i$'s actions, he can achieve the same learning speed as $i$.  This implies that his posterior likelihood increases fast. Hence, as this outside observer's posterior likelihood coincides with agent $i$'s social likelihood, the precision of $i$'s social information increases at a speed of at least $r$.

Next, we focus on the second component of $L^{i,\fg,\ff}_t$: agent $i$'s private likelihood, $P_t^{i,\fg,\ff}$. As agents receive more independent private signals over time, their private information about the state becomes more precise. However, the precision of their private information cannot increase without bound, as shown in the following lemma. 

\begin{lemma}\label{lemma:private_signal}
For any agent $i$ at time $t$ and any distinct $\fg,\ff \in \cS$, the absolute value of the private likelihood is at  most $t \cdot M$, i.e. \[
\frac{1}{t}|P_t^{i,\fg,\ff}|\leq M \quad\text{almost surely.}
\]
\end{lemma}

This lemma states that that at any given time $t$, there is an upper bound to the precision of agents' private information, which only depends on the private signals distribution and is independent of the structure of the network and the history of observed actions. Notice that since $P^{i,\fg,\ff}_t = L^{i,\fg,\ff}_t - S^{i,\fg,\ff}_t$ by  \eqref{eq:beliefs}, it  captures the difference between what $i$ knows about the state and what an outside observer who observes $i$'s actions and his neighbors' actions would know about the state. Thus, the bound $M t$ assigned to $i$'s private signals also applies to the difference between the posterior likelihood $L^i_t$ and the social likelihood $S^i_t$.

\subsection{Proof sketch for Theorem~\ref{main theorem}}

We end this section by providing a sketch of the proof of Theorem~\ref{main theorem} using our earlier results. Suppose to the contrary that in equilibrium, agents learn at a speed that is strictly higher than $M$, where $M$ is twice the log-likelihood ratio of the strongest signal. Then, by Lemma~\ref{prop1}, the social information would become precise at a speed that is also strictly higher than $M$. Meanwhile, at any given time $t$, the precision of the private information is at most $Mt$, as shown in Lemma~\ref{lemma:private_signal}. Hence, by  \eqref{eq:beliefs}
the sign of $L_t^i$ would be determined purely by the social information after some (random) time.
By Lemma~\ref{lemma:strategic-myopic}, there exists a time $T$ such that from $T$ onward, in equilibrium, all agents act only based on the social information, and furthermore they would choose the myopic action even though they are forward-looking. Consequently, their actions would no longer reveal any information about their private signals and information would cease to be aggregated. This contradicts our hypothesis that the precision of the social information grows at such a high speed. Therefore, we conclude that the equilibrium speed of learning in networks does not increase beyond $M$.

\section{Networks which are not Strongly Connected}\label{sec:not-strongly-connected}
So far, we have focused on strongly connected networks where there is an observational path between every pair of agents. While on a strongly connected network all agents learn at the same speed, this is not true for general networks. For example, consider a simple star network where there is a single agent at the center who observes everyone, and where the remaining peripheral agents observe no one. Here, the peripheral agents' actions are independent conditional on the state. These actions supply the central agent with $n-1$ additional independent signals, and he thus learns at a speed that increases linearly with the number of agents.\footnote{See Theorem 5 in \cite{harel2021rational}.} In contrast, all peripheral agents learn at a constant speed $r_a$, as in the single-agent case. Hence, for general networks, depending on the structure of the network, some agents can learn faster than others.

More importantly, this star network example implies that the bound obtained in Theorem~\ref{main theorem} does not hold for all agents in a non-strongly connected network. The intuitive reason is that when the network is not strongly connected, some agents might remain unobservable to others, e.g., the central agent in the star network. These agents can thus learn very fast from observing others since their own past actions do not affect the actions of others, rendering others' past actions more informative. This cannot happen in a strongly connected network where every agent (potentially indirectly) learns about the actions of every other agent.

Nevertheless, even though it is not necessarily true that all agents learn slowly, our next result establishes that it is still true that some agents will learn slowly in any network.

\begin{proposition}\label{prop:bound-minimal-learning-speed}
Consider an arbitrary (finite) network and let $r_i$ be the speed of learning of agent $i$.
We have that $\min_i r_i \leq M$.
\end{proposition}
The proof of Proposition \ref{prop:bound-minimal-learning-speed} relies on the idea that within any general network, there is always a  strongly connected sub-network, say $E$, in which no agent observes any agent outside of this sub-network. Thus, the learning process  at  $E$ is independent of the agents outside of $E$. Since $E$ is strongly connected, Theorem \ref{main theorem} implies that the speed of learning on $E$ is bounded by $M$.

\subsection{Learning on a Line}

In this section we discuss the commonly studied case of an infinite group of agents who learn by observing (a subset) of their predecessors. A prominent feature of the line network is that information is transmitted unidirectionally. This simplified observational structure has thus received particular attention in the herding literature where agents act only once. Here, we extend it to our setting where agents act repeatedly.  As we will see below, the arguments we use on line networks are reminiscent of the sequential social learning literature, except that we focus on the speed of learning rather than whether learning occurs or not.


We first consider the case where agents observe a general subset of their predecessors, which in the herding literature was considered in \cite{acemoglu2011bayesian}.\footnote{Their conclusion is qualitatively different: For some network structures asymptotic learning obtains, and for some it does not. Of course, their notion of learning (namely that agent $i$ takes the correct action with probability that tends to 1 as $i$ tends to infinity) is very different than ours.}

\begin{proposition} \label{prop:inf_nocycle}
Suppose that $N=\{1,2,\ldots\}$ and $N_i \subseteq \{1,\ldots,i\}$.
Then there is some constant $K$ such that the speed of learning $r_i \leq K$ for infinitely many agents $i$.
\end{proposition}
That is, under these assumptions, it is impossible that the speed of learning $r_i$ tends to infinity with $i$. Thus, even though there are infinitely many agents, many agents will have a low speed of learning. 
This result thus shows that in the line networks, the conclusion of Proposition~\ref{prop:bound-minimal-learning-speed} that the speed of learning of some agent must be bounded generalizes to the conclusion that the speed of an infinite number of agents must be bounded.

The proof of this proposition uses ideas that are similar to those of the proof of Theorem~\ref{main theorem}.
Since observations are unidirectional, agents behave myopically even for $\delta>0$.
To show the result, suppose towards a contradiction, that there is no such $K$. Then, in particular, there are only finitely many agents whose speed of learning is less than $M$. Each of the remaining agents eventually stops using their private signals, because the fact that they learn so quickly means that they observe very strong social information. It follows that the only information that is aggregated asymptotically is that of the finite group of agents who learn slowly. Thus, it is impossible that any agent has a high speed of learning, which is a contradiction.

We now consider the special case of $N_i=\{i-1,i\}$, i.e., each agent observes only their direct predecessor.

\begin{proposition}\label{prop:chain}
    When $N_i=\{i-1,i\}$,  $r_i \leq M$ for all agents $i$.
\end{proposition}
Thus, the conclusion of Theorem~\ref{main theorem} applies also to this case of a non-strongly connected network.
Again, we prove Proposition \ref{prop:chain} by using the ideas behind Theorem \ref{main theorem}. Notice first that the imitation principle for myopic agents implies that the speed of learning is weakly increasing in the index  of the agent. Now, suppose to the contrary that there exists some agent $k$ who learns at a speed that is strictly greater than $M$. Then all agents $j>k$ would also learn at a high speed. Eventually, all these high-speed learning agents would stop using their private signals because the fact that they learn so fast means that they observe very precise social information. Consequently, information aggregation stagnates and thus learning cannot be too fast.

\section{Conclusion}

In this paper, we show that information aggregation is highly inefficient for large groups of agents who learn from private signals and by observing their social network neighbors.
To overcome the difficulty of constructing equilibria explicitly, we focus on the asymptotic speed of learning, allowing us to prove results that apply to all equilibria.
We show that regardless of the size of the network, the speed of learning is bounded above by a constant, which only depends on the private signal distribution (and not on the discount factor, the observational graph, the agents' utilities, or prior belief).

An important limitation of our results is that they only apply asymptotically. As our numerical results show, these asymptotic results can apply already from the early periods, for myopic agents on particular networks. However, for patient agents, it is unclear whether it takes a long time for the asymptotic results to apply. Calculating welfare for patient agents seems beyond what is currently tractable, as it would require a detailed analysis of the equilibria of this game. In fact, even for two myopic agents who observe each other, it seems intractable to calculate welfare, and moreover, it remains unknown whether the asymptotic speed of learning is strictly greater than that of one agent who learns on his own. Nevertheless, even in the most general settings of patient agents on complex networks, it seems reasonable to conjecture that the main economic force driving our result\textemdash that learning cannot be too fast because it would lead agents to disregard their own signals\textemdash is  significant even in the early periods.  We leave this for future research. 

Another promising direction for future research is the calculation of lower bounds for the speed of learning. Currently, we cannot show that equilibrium speed of learning on any connected network is faster than that of a single agent. Even without imposing equilibrium, this question remains open: What learning speed can be achieved when a social planner is allowed to choose the agents' strategies? For the complete network, it is still unknown how a social planner could achieve any speed of learning that is better than the speed achieved by a single agent. The challenge for  the social planner lies in the trade-off between using the actions to communicate between the agents and choosing the correct actions with very high probability at the same time. One conjecture from \cite{harel2021rational} is that a better speed can be achieved by having the social planner instruct the agents to behave as if they are myopic and over-weight their own signals, causing their actions to reveal more information. In simpler, sequential settings, this type of mechanism was shown to indeed improve learning outcomes \citep{arieli2023hazards}.

It is possible to extend our model in a number of directions. In Appendix~\ref{app:random-time} we show that our main result continues to hold if agents do not observe each other every period, but only in some (potentially random) periods. A natural extension, which we leave for future work, is to allow the network to be random and its realization to be only partially observed by the agents. An interesting technical question is that of the robustness of our results to the assumption of bounded private signals. We conjecture that a result similar to our main theorem should hold even if signals are unbounded, with the Kullback-Leibler divergence between the conditional signal distributions playing the bounding role currently played by the maximum log-likelihood ratio. This is indeed the case for myopic agents on the complete network \citep{harel2021rational}. A substantive extension is to allow the underlying state to change over time \citep[see][]{moscarini1998social, frongillo2011social, dasaratha2020learning}. For example, the underlying unknown state could capture the quality of a local restaurant or school, which might fluctuate gradually. In such a setting, one could replace the speed of learning metric with the long-run probability of making the correct choice and study whether information gets aggregated in this case and, if so, whether the information aggregation process is efficient.




\appendix
\section{Proofs}
 \begin{proof}[Proof of Lemma \ref{lemma:strategic-myopic}]
(i) First notice that we can write the agent's expected utility conditioned on his information $I^i_t$ at time $t$ as the sum $U_{<t}+\delta^{t} U_{\geq t}$, where
\begin{align*}
    U_{<t}=(1-\delta)\sum_{k=1}^{t-1}\delta^{k-1} \E[u(a^i_{k},\omega)|I^i_t]
\end{align*}
is the sum of the expected flow utilities until time $t$, and $U_{\geq t}$ is the expected continuation utility at time $t$ given by
\begin{align*}
    U_{\geq t} 
    &= (1-\delta)\sum_{k=0}^\infty \delta^{k} \E[u(a^i_{t+k},\omega)|I^i_t]\\ 
    &= (1-\delta)\E[u(a^i_t,\omega)|I^i_t] + \delta(1-\delta)\sum_{k=0}^\infty \delta^{k} \E[u(a^i_{t+k+1},\omega)|I^i_t].
\end{align*}

Fix some state $\fg \in \cS$. Since $a^\fg$ is the unique optimal action in state $\fg$, by applying an affine transformation to the flow utility function $u\colon \AA\times\cS \to \R$ we can assume that $u(a^\fg,\fg)=1$, that $u(a,\fg) \leq 0$ for all $a \neq a^\fg$. Let $c^\fg =\max_{a,\ff}|u(a,\ff)|$. Recall that $p^{i,\fg}_t=\Prob[\omega=\fg|I^i_t]$ is the agent's posterior at time $t$. Thus, for any action $a \neq a^\fg$, since $u(a,\fg) \leq 0$, the expected flow utility $\E[u(a,\omega)|I^i_t] = \sum_{\ff \in \cS} u(a, \ff) \cdot p^{i,\ff}_t$ is at most $c^{\fg} (1-p^{i,\fg}_t)$. Likewise, $\E[u(a^i_t,\omega)|I^i_t]$ is at most $p^{i,\fg}_t+c^{\fg}(1-p^{i,\fg}_t)$ since $u(a^{\fg}, \fg) =1$.

Now, suppose that $a_t^i = a \neq a^\fg$. Then the expected continuation utility is
\begin{align*}
    U_{\geq t} 
    &= (1-\delta)\E[u(a,\omega)|I^i_t] + \delta(1-\delta)\sum_{k=0}^\infty \delta^{k} \E[u(a^i_{t+k+1},\omega)|I^i_t]\\
    &\leq (1-\delta)c^{\fg} (1-p^{i,\fg}_t) + \delta(p^{i,\fg}_t+c^{\fg}(1-p^{i,\fg}_t))\\
    &= \delta p^{i,\fg}_t+c^{\fg}(1-p^{i,\fg}_t).
\end{align*}
On the other hand, the strategy that chooses $a^\fg$ from period $t$ onward has an expected continuation utility at least  $p^{i,\fg}_t  -c^{\fg} (1-p^{i,\fg}_t)$. Thus, when
\begin{align}
\label{eq:myopic-strategic}
    p^{i,\fg}_t  -c^\fg (1-p^{i,\fg}_t) > \delta p^{i,\fg}_t+c^{\fg}(1-p^{i,\fg}_t)
\end{align}
the agent cannot choose $a^i_t= a \neq a^\fg$ in period $t$. Rearranging, this happens when
\begin{align*}
    p^{i,\fg}_t > \frac{2c^\fg}{2c^\fg+1-\delta}.
\end{align*}
Thus, under the above condition, agents choose $a^i_t=a^\fg$ in equilibrium. Clearly, the myopic action is then also equal to $a^\fg$, as this corresponds to the case $\delta=0$.
Part (i) of the lemma now follows by setting $c=\max_\fg 2c^\fg$.

For part (ii), fix a discount factor $\delta\in [0, 1)$ and let $\sigma$ be an equilibrium. Let $a^i_t$ be the action taken by $i$ at time $t$ under $\sigma$. Since the entire sequence of private signals reveals the state, conditional on $\omega=\fg$, $\lim_t p^{i,\fg}_t=1$ almost surely. Hence, 
by part (i) the agent will choose the unique optimal action that is also myopically optimal in equilibrium from some (random) time on. Since there are finitely many agents, this will hold for all $i$ for all $t$ larger than some (random) $T$. This means that from $T$ onward, in equilibrium all agents will behave myopically.
\end{proof}

The following is a simple consequence of Lemma~\ref{lemma:strategic-myopic}.
\begin{corollary}
\label{cor:strategic-myopic}
There is a constant $C>0$ such that, in any equilibrium, if it holds for $\fg \in \cS$ and all $\ff \neq \fg$ that $L^{i,\fg,\ff}_t > C$, then $a^i_t = m^i_t = a^\fg$.
\end{corollary}
\begin{proof}
Fix a state $\fg \in \cS$. Suppose that $L^{i,\fg,\ff}_t > C$ for some $C>0$ to be chosen later. Then $p^{i,\fg}_t > \ee^C p^{i,\ff}_t$. If this holds for all $\ff \neq \fg$ then 
\begin{align*}
    p^{i,\fg}_t > \frac{1}{|\cS|-1}\ee^C \sum_{\ff \neq \fg}p^{i,\ff}_t = \frac{1}{|\cS|-1}\ee^C (1-p^{i,\fg}).
\end{align*}
Thus for any each $c>0$, $u$ and $\delta \in [0,1)$, for $C$ large enough it holds that
\begin{align*}
    p^{i,\fg}_t > \frac{c}{c+1-\delta},
\end{align*}
and the result follows by part (i) of Lemma~\ref{lemma:strategic-myopic}.
\end{proof}

The following lemma will be useful in the proofs below. Recall that we denote $\bar{u} = \E[u(a^\omega,\omega)]$.
\begin{claim}
\label{lem:W}
     There exist $\underline{c},\overline{c}>0$ such that
    \begin{align*}
 \underline{c}\cdot\Prob[a^i_t \neq a^\omega] \leq        \bar{u}-\E[u(a^i_{t},\omega)] \leq \overline{c}\cdot\Prob[a^i_t \neq a^\omega]
    \end{align*}
\end{claim}
Recall that we can think of the difference $\bar{u} - \E[u(a^i_t,\omega)]$ as the expected loss in flow utility as compared to the first-best. The lemma above states that this quantity is the same---up to constants---as the probability of choosing the correct action. An immediate consequence of this lemma is that we can express the speed of learning in terms of this loss:
\begin{align}
\label{eq:r-w}
    r_i = \liminf_t -\frac{1}{t}\log\Prob[a^i_t \neq a^\omega] = \liminf_t -\frac{1}{t}\log(\bar{u}-\E[w(a^i_t, a^\omega)]).
\end{align}
\begin{proof}[Proof of Claim~\ref{lem:W}]
Denote by $\underline{c}$ the minimum loss of utility from choosing incorrectly:
\begin{align*}
    \underline{c} = \min_{\fg \in \cS, a \neq a^\fg}u(a^\fg,\fg)-u(a,\fg).
\end{align*}
Since there is a unique optimal action in each state we have that $\underline{c}>0$. Analogously, denote by $\overline{c}>0$ the maximum such loss:
\begin{align*}
    \overline{c} = \max_{\fg \in \cS, a \neq a^\fg}u(a^\fg,\fg)-u(a,\fg).
\end{align*}
Then
\begin{align*}
    \bar{u}-\E[u(a^i_{t},\omega)] 
    = \E[u(a^\omega,\omega)-u(a^i_{t},\omega)] 
    \leq \E[\overline{c}\cdot 1_{a^i_t \neq a^\omega}]
    = \overline{c}\cdot\Prob[a^i_t \neq a^\omega].
\end{align*}
Likewise,
\begin{align*}
    \bar{u}-\E[u(a^i_{t},\omega)] 
    = \E[u(a^\omega,\omega)-u(a^i_{t},\omega)] 
    \geq \E[\underline{c}\cdot 1_{a^i_t \neq a^\omega}]
    = \underline{c}\cdot\Prob[a^i_t \neq a^\omega].
\end{align*}
\end{proof}

\begin{proof}[Proof of Lemma \ref{lem:rate-bound-from-observing}]

Suppose $i$ observes $j$. Let $\sigma$ be an equilibrium and let $a^i_t$ be the action taken by $i$ at time $t$ under $\sigma$. We claim that for $t>1$, 
\begin{align}
\label{eq:imitation}
(1-\delta)(\bar{u}-\E[u(a_t^i,\omega)]) \leq \bar{u}-\E[u(a_{t-1}^j, \omega)].
\end{align}
To see that this equation must hold observe that the left-hand side equals the expected continuation loss the agent would have from time $t$ on if he chooses the action $a_t^i$ at time $t$ and suffered no loss in future periods. Hence this is smaller than the expected continuation loss under the strategy profile $\sigma$.  In equilibrium, this must be smaller than the loss from any deviation, and the right-hand-side equals the loss the agent suffers when imitating $j$'s action $a_{t-1}^j$ from time $t$ onward.
Thus the above inequality must hold.

As a consequence, 
\begin{align*}
   \liminf_{t \to \infty} -\frac{1}{t} \log (\bar{u}-\E[u(a_t^i,\omega)])
    &\geq \liminf_{t \to \infty} -\frac{1}{t} \log\left(\frac{1}{1-\delta} (\bar{u}-\E[u(a_{t-1}^j,\omega)])\right)\\  
    &=\liminf_{t \to \infty} -\frac{1}{t} \log (\bar{u}-\E[u(a_{t-1}^j,\omega)])\\
    &=\liminf_{t \to \infty} -\frac{1}{t} \log (\bar{u}-\E[u(a_{t}^j,\omega)]).
\end{align*}
Thus part (i) follows from \eqref{eq:r-w}. Part (ii) follows as in any strongly connected network, there is an observational path from each agent $i$ to each other agent $j$ and the monotonicity of learning speed shown in (i) applied along this path, implies that $r_i \geq r_j$. The opposite inequality holds by the same argument. 
\end{proof}

The next simple claim will be helpful in the Proof of Lemma \ref{prop1}.
\begin{claim}
\label{lemma:markov}
     Suppose that $X_t$ is a sequence of random variables taking values in $[0,1]$ such that $\lim_t -\frac{1}{t}\log\E[X_t|\Sigma_t] \geq r$ almost surely, for some  sequence of sigma-algebras $\Sigma_t$.Then $\liminf_t -\frac{1}{t}\log X_t \geq r$, almost surely.
\end{claim}
\begin{proof}[Proof of Claim \ref{lemma:markov}]
By the claim hypothesis there exists a random $F\colon \mathbb{N} \to \R$, such that $\lim_t F(t)/t=0$ almost surely and such that 
\[
  \E[X_t|\Sigma_t] \leq \ee^{-r t +F(t)}.
\]
We can furthermore assume that $F(t) \leq rt$, since $X_t \leq 1$. Taking expectations of both sides yields
\[
  \E[X_t] \leq \ee^{-r t}\cdot\E[\ee^{F(t)}].
\]
Hence, if we denote $f(t) = \log\E[\ee^{F(t)}]$, 
\[
  \E[X_t] \leq \ee^{-r t+f(t)}.
\]
Furthermore,
\begin{align*}
  \lim_t \frac{1}{t}f(t) 
  = \lim_t \frac{1}{t}\log\E[\ee^{F(t)}]
  = \lim_t \log\E[\ee^{F(t)/t}]
  = 0
\end{align*}
where the last equality is a consequence of the facts that $F(t)/t \leq r$ and $F(t)/t$ converges almost surely to zero.
Hence, by Markov's inequality, for any $c_t > 0$,
\[
    \Prob[X_t \geq c_t] \leq \E[X_t]/c_t \leq \ee^{-rt + f(t)}/c_t\,.
\]
Choosing $c_t=\ee^{-rt+f(t)}\cdot t^2$, we get that
\[
    \Prob[X_t \geq \ee^{-rt+f(t)+2\log t}] \leq \frac{1}{t^2}\,.
\]
Hence, by Borel-Cantelli, almost surely $X_t \leq \ee^{-rt+f(t)+2\log t}$ for all $t$ large enough, and in particular, $\liminf_t-\frac{1}{t}\log X_t \geq r$.
\end{proof}

\begin{proof}[Proof of Lemma \ref{prop1}]
In this proof we use Landau notation, so that $o(t)$ stands for some function $f\colon \mathbb{N} \to \mathbb{R}$ such that $\lim_t f(t)/t=0$.

By assumption and the definition of speed of learning in \eqref{def:speed}, $\Prob[a^i_t \neq a^\omega] \leq  \ee^{-rt+o(t)}$.
Let
\begin{align*}
    p^{x,\ff}_t = \Prob[\omega=\ff|a^i_t]
\end{align*}
be the probability assigned to $\omega=\ff$ by an outside observer $x$ that sees only agent $i$'s action at time $t$.

By Bayes' Law,
\begin{align*}
    \Prob[ \omega=\ff | a^i_t=a^\fg] = \frac{\Prob[a^i_t=a^\fg , \omega=\ff]}{\Prob[a^i_t=a^\fg]}.
\end{align*}
Since $\Prob[a^i_t \neq a^\omega] \leq \ee^{-rt+o(t)}$, we can bound the denominator by 
\[
\Prob[a^i_t=a^\fg] \geq \Prob[\omega=\fg,a^i_t=a^\omega] \geq \Prob[\omega=\fg]-\Prob[a^i_t \neq a^\omega] \geq \Prob[\omega=\fg]-\ee^{-rt+o(t)}.
\]
If $\fg \neq \ff$ then the numerator is at most $\Prob[a^i_t \neq a^\omega]$ since the event that the agent takes the wrong action contains the event that the agent takes action $a^\fg$ in state $\ff$. Hence
\begin{align*}
    \Prob[ \omega=\ff | a^i_t=a^\fg] \leq \frac{\ee^{-rt+o(t)}}{\Prob[\omega=\fg]-\ee^{-rt+o(t)}}.
\end{align*}
Now, because $\Prob[a^i_t \neq a^\omega] \leq \ee^{-rt+o(t)}$, by Borel-Cantelli, almost surely $a^i_t = a^\fg$ for all $t$ large enough, conditioned on $\omega=\fg$. It follows that  for all $t$ large enough---again conditioned on $\omega=\fg$---the belief $p^{x,\ff}_t$ will equal $\Prob[ \omega=\ff | a^i_t=a^\fg]$. Hence, by the displayed equation above, $\liminf_t -\frac{1}{t}\log p^{x,\ff}_t \geq r$. Since this holds for all $\ff \neq \fg$, we get that $\lim_t p^{x,\fg}_t=1$. 

Now, let
\begin{align*}
    p^{y,\ff}_t = \Prob[\omega=\ff | H^i_t]
\end{align*}
be the probability assigned to $\omega=\ff$ by an outside observer $y$ that sees only agent $i$'s public history $H_t^i$.

Then
\[
\log \frac{p^{y,\fg}_t}{p^{y,\ff}_t} = \log\frac{\Prob[H^i_t | \omega =\fg]}{\Prob[H^i_t | \omega =\ff]} + \log \frac{\Prob[\omega= \fg]}{\Prob[\omega =\ff]} = Q^{i,\fg,\ff}_t.
\]
Since $H^i_t$ includes $a^i_{t-1}$, the law of total expectations yields that
\begin{align*}
    p^{x,\ff}_{t-1} = \E[p^{y,\ff}_{t}|a^i_{t-1}].
\end{align*}
It now follows from Claim~\ref{lemma:markov} that identical asymptotics apply to $p^{y,\ff}_t$: $\liminf_t -\frac{1}{t}\log p^{y,\ff}_t \geq r$ and $\lim_t p^{y,\fg}_t=1$. Thus, conditioned on $\omega=\fg$,
\begin{align*}
    \liminf_{t \to \infty} \frac{1}{t}Q^{i,\fg,\ff}_t = \liminf_{t \to \infty}\frac{1}{t}\log \frac{p^{y,\fg}_t}{p^{y,\ff}_t} \geq r \quad \text{almost surely}.
\end{align*}

\end{proof}

\begin{proof}[Proof of Lemma \ref{lemma:private_signal}]
Recall that at time $t$, $H^i_t = \{a_s^j \colon s<t,j \in N_i\}$ is the history of actions observed by $i$ and $s^i_1,\ldots,s^i_t$ is the sequence of private signals received by $i$. Given a \emph{pure} strategy profile $\sigma$, agent $i$ chooses a unique action $a^i_t \in \AA$ at time $t$: $a^i_t = \sigma(s^i_1, \cdots, s^i_t, H^i_t)$.
It follows that for each history $H^i_t$ there is a set $\sig^{i}(H^i_t) \subseteq \sig^i_{1} \times \dots \times \sig^i_{t-1}$ of possible private signal realizations $s^i_1,\ldots,s^i_{t-1}$ that are consistent with $ H^i_t$:
\[
    \sig^{i}(H^i_t) = \left\{ \ss_1^i,\ldots,\ss_{t-1}^i \in \sig^i_{1} \times \dots \times \sig^i_{t-1}\,:\, \Prob[s_1^i=\ss_1^i,\ldots,s_{t-1}^i=\ss_{t-1}^i|H_t^i] > 0
    \right\} \,.
\]
In other words, if we imagine an outside observer who sees only $H^i_t$---i.e.\ sees $i$'s actions and his neighbors' actions---then $\sig^{i}(H^i_t)$ is the set of private signal realizations 
$(s^i_1,\ldots,s^i_{t-1})$ to which this observer assigns positive probability.

Consider the numerator $\Prob[s_{1}^i,\ldots,s_t^i|H^i_t,\omega=\fg]$ of the private log-likelihood ratio $P^{i,\fg,\ff}_t$. Using the definition of $\sig^{i}(H^i_t)$, we can write
\[
    \Prob[s_{1}^i,\ldots,s_t^i|H^i_t,\omega=\fg] = \Prob[s^i_1,\ldots,s^i_t|(s^i_1,\ldots,s^i_{t-1}) \in \sig^{i}(H^i_t),\omega=\fg] \quad \text{almost surely.}
\]
The above equality holds as conditional on $\omega=\fg$ the signals of different agents are independent and hence the only relevant information about agent $i$'s signals $s_1^i,\ldots,s_t^i$ contained in the history $H_t^i$ is the restriction the history imposes on the realization of these signals. 

Let $\mu_{1 \ldots t}^{i,\fg}$ be the measure over signal realizations $s_1^i,\ldots,s_t^i$ when $\omega = \fg$. Then
\begin{align*}
\Prob[s^i_1,\ldots,s^i_t|(s^i_1,\ldots,s^i_{t-1}) \in \sig^{i}(H^i_t),\omega=\fg] = \frac{\mu^{i, \fg}_{1 \ldots t}(s^i_1,\ldots,s^i_t)}{\mu^{i, \fg}_{1 \ldots t-1}(\sig^{i}(H^i_t))}.
\end{align*} 
We thus have that 
\begin{align}
  \label{eq:signal-llr-transformed}
    P_t^{i,\fg,\ff} =\log \frac{\mu^{i, \fg}_{1 \ldots t}(s^i_1,\ldots,s^i_t)}{\mu^{i,\ff }_{1 \ldots t}(s^i_1,\ldots,s^i_t)} + \log\frac{\mu^{i, \ff}_{1 \ldots t-1}(\sig^{i}(H^i_t))}{\mu^{i, \fg}_{1 \ldots t-1}(\sig^{i}(H^i_t))}.
\end{align}
Since the signals are independent over time the first term of \eqref{eq:signal-llr-transformed} is equal to 
\begin{align*}\sum_{\tau=1}^t\log  \frac{\mu^{i,\fg}_{\tau}(s^i_\tau)}{\mu^{i,\ff}_{\tau}(s^i_\tau)},
\end{align*}
which is at most $\frac{1}{2}Mt$. The second term of \eqref{eq:signal-llr-transformed}, which is equal to 
\[
\log \frac{\sum_{(\ss_{1}^i, \cdots, \ss_{t-1}^i) \in \sig^{i}(H^i_t)} \prod_{\tau =1}^{t-1} \mu^{i,\ff}_{\tau}(\ss_{\tau}^i)}{\sum_{(\ss_{1}^i, \cdots, \ss_{t-1}^i) \in \sig^{i}(H^i_t)} \prod_{\tau =1}^{t-1} \mu^{i,\fg}_{\tau}(\ss_{\tau}^i)}
\]
is also at most $\frac{1}{2}M(t-1)$.\footnote{This follows from the fact that for any two sequences of positive numbers $(a_1,\ldots,a_n)$ and $(b_1,\ldots,b_n)$ it holds that 
\[
    \frac{\sum_k a_k}{\sum_k b_k} = \frac{\sum_k b_k (a_k/b_k)}{\sum_k b_k} \in \left[\min_k (a_k/b_k), \max_k (a_k/b_k)\right] \,.
\]} Thus, it follows that $P^{i,\fg,\ff}_t$ is at most $Mt$. By an analogous argument $P^{i,\ff,\fg}_t=-P^{i,\fg,\ff}_t$ is at least $Mt$, and so $|P^{i,\fg,\ff}_t|$ is at most $Mt$.
\end{proof}

\begin{proof}[Proof of Theorem \ref{main theorem}]
Fix a discount factor $\delta \in [0,1)$. Let $\sigma$ be an equilibrium and $a^i_t$ be the action taken by $i$ at time $t$ under $\sigma$. Now consider an outside observer $x$ who observes everybody's actions so that the information available to him at time $t$ is $H_t = \{a^i_s, i \in N, s\leq t\}$ and at time infinity is $H_{\infty}= \cup_{t} H_t$. Thus at any time $t$, this observer can calculate the social likelihood $Q_t^i$ for all $i$.

Suppose that the social likelihood is high, and in particular $Q_t^{i,\fg,\ff} > Mt+C$ at some $t$ for some constant $C$, some state $\fg$ and all $\ff \neq \fg$. Since the private likelihood $P^{i,\fg,\ff}_t$ cannot be less than $-Mt$ (Lemma \ref{lemma:private_signal}), the posterior likelihood $L_t^{i,\fg,\ff} = Q_t^{i,\fg,\ff}+P_t^{i,\fg,\ff} > C$. In this case, supposing $C$ is high enough, by Corollary~\ref{cor:strategic-myopic}, the agent will choose the myopic action $a_t^i = m_t^i = a^\fg$.  Thus, under this condition on $Q^i_t$ the outside observer will know which action the agent will choose in equilibrium, and will not learn anything (in particular, about the agent's signals  or the state) from observing this action.

Suppose towards a contradiction that the equilibrium speed of learning $r$ is strictly higher than $M$, i.e.  $r = M+\varepsilon$ for some $\varepsilon>0$. Then it follows from Lemma \ref{prop1} that for any $C>0$, $Q^{i,\fg,\ff}_t \geq  (M+\varepsilon)t > Mt+C$ for all $t$ large enough and all $\ff \neq \fg$. Since there are finitely many agents, this will hold for all $i$, for all $t$ larger than some (random) $T$. Hence the outside observer $x$ learns nothing more from the agents' actions after time $T$.

Let $a^x_t$ be the action that $x$ would choose to maximize the probability of matching the state at time $t$. Since no new information is gained after time $T$, the outside observer stops updating their action and so $a^x_T = a^x_\infty$. Hence $\Prob[a^x_\infty \neq a^\omega] > 0$.

Since $x$ observes everyone's actions, by the imitation principle
\begin{align}
\label{eq:x-infty}
 \Prob[a^i_t \neq a^\omega] \geq \Prob[a^x_\infty \neq a^\omega] > 0
\end{align}
for all agents $i$ and all times $t$. But since the equilibrium speed of learning is $M+\varepsilon > 0$, by the definition of speed of learning in \eqref{def:speed}, $\Prob[a^i_t \neq a^\omega]$ converges to zero, in contradiction with \eqref{eq:x-infty}.

\end{proof}

\begin{proof}[Proof of Proposition \ref{prop:bound-minimal-learning-speed}]

Recall that $E \subseteq N$ is a strongly  connected component if there is an indirect observation path from each agent in $E$ to each other agent in $E$. By a standard argument\footnote{Define a preorder on the set of agents $N$ by $i \succeq j$ if there is an indirect observation path from $i$ to $j$. The $\succeq$-equivalence classes are the strongly connected components. The set $E$ is any equivalence class of $\succeq$-minimal elements.}, there exists a strongly connected component $E$ in which no agent observes agents outside of $E$. Hence we can analyze the speed of learning of agents in $E$ in isolation and apply Theorem~\ref{main theorem} to conclude that $r_i \leq M$ for all $i \in E$.
\end{proof}

\begin{proof}[Proof of Proposition \ref{prop:inf_nocycle}] Suppose towards a contradiction that there does not exist a constant $K$ where $r_i \leq  K$ for infinitely many agents, i.e., that $\lim_i r_i=\infty$. This implies that there exists a finite $k$ such that the agents with $r_i \leq M$ constitute a subset of $\{1, \ldots, k\}$. Thus, for all agents $i >k$, $r_i > M$. Fix an agent $n >k$. 

Consider an outside observer $x$ who observes agents $\{1, \ldots, n\}$. Since $r_i >M$ for all $i \in \{k+1, \ldots, n\}$, it follows from Lemma \ref{prop1} that conditioned on $\omega=\fg$, $Q^{i,\fg,\ff}_t > M \cdot t$ for all $i \in \{k+1, \ldots, n\}$, for all $\ff \neq \fg$ and for all $t$ large enough. Since the absolute value of the private likelihood $P^{i,\fg,\ff}_t$ is always less than or equal to $Mt$ at any given time $t$ (Lemma \ref{lemma:private_signal}), all agents $k+1, \ldots, n$ will eventually ignore their private information and act based on their social information. I.e., for all $i\in \{k+1, \ldots, n\}$, $a^i_t$ is determined by $Q^i_t$ for all $t$ large enough. Since there are finitely many agents observed by $x$, there exists a (random) $T$ so that from $T$ onward, $x$, who knows $Q^i_t$, learns nothing more from the actions of agents $\{k+1, \ldots, n\}$. Thus, $x$ will learn as fast as he would if he only observed agents $\{1, \ldots, k\}$. It  follows that the speed of learning of this outsider observer $x$ is at most $k r_a$, which is $k$ times the speed of learning of a single agent, or, equivalently, the speed of learning from directly observing $k$ signals every period. Following the argument in the proof of Theorem~\ref{main theorem}, since $x$ observes $i$, $x$ learns at least as fast as $i$, and so $r_i \leq kM$. Since this holds for all $i$, we have reached a contradiction to the assumption that $\lim_i r_i = \infty$.
\end{proof}

\begin{proof}[Proof of Proposition \ref{prop:chain}]
By the imitation principle for myopic agents, $$r_{i+1} \geq r_{i}$$ for all $i = 1, 2, \ldots$ and $r_1 = r_a < M$ as agent $1$ sees only their private signals.
Suppose towards a contradiction that there exits an agent $k  \geq  2$ agent who learns at a speed that is strictly higher than $M$, i.e., $r_k = M + 2\eps$ for some $\eps >0$ and let $k>1$ be the smallest such integer.  
By Lemma~\ref{prop1} there exists a (random) $T$ such that for all $t$ larger than $T$, in state $\fg$ the public log-likelihood is greater that the largest likelihood than can be induced by any private signal
\begin{equation}\label{eq:social}
    Q^{k,\fg,\ff}_t > M \cdot t.    
\end{equation}

Consider an outside observer $x$ who observes $H^k_t$: the actions of agents $k-1$ and $k$. Hence $x$ knows $Q^{k,\fg,\ff}_t$ for all $t >T$, and by the same argument of the proof of Theorem~\ref{main theorem}, does not learn anything from $k$'s action after time $T$. Hence $x$'s speed of learning is $r_{k-1}$. As the outside observer $x$ can observe $k$, by the imitation principle, $x$'s speed of learning is at least that of agent $k$. This together implies that $r_{k-1} > M$, which is a contradiction to the original assumption on $k$. 
\end{proof}

\section{Intermittent and Random Observation Times} \label{app:random-time}

In this section we extend our model to allow for some intermittent and random observation times. This extension allows us to consider, for example, a situation in which one pair of agents meet every day, another pair meets every Sunday, and yet another pair meets on a random day of the week.  Our main  result still applies in this setting, and moreover the same proofs apply, with some additional details that need to be verified, as we explain.

Formally, for each pair of agents $i,j$ such that $j \in N_i$ let $O_{i,j}$ be the set of time periods in which $i$ observes $j$. These sets can be random, but we assume that they are independent of each other, the state and the signals. We also assume that there is some number $D > 0$ such that, with probability $1$, for every $i,j$ such that $j 
\in N_i$ and every $t \in \{0,1,2,\ldots\}$, the intersection $\{t+1,\ldots,t+D\} \cap O_{i, j}$ is not empty. That is, if $j \in N_i$ then $i$ observes $j$ at least once every $D$ periods. Hence, the difference between consecutive $t_1,t_2 \in O_{i, j}$ is at most $D$. 

The history of actions observed by agent $i$ at time $t$ is
\begin{align*}
    H^i_t = \{a^j_s\,:\, s < t, s \in O_{i,j}, j \in N_i\}.
\end{align*}
As before, the private history of agent $i$ at time $t$ is $I^i_t = (s^i_1,\ldots,s^i_t,H^i_t)$, and a pure strategy at time $t$ is a map that assigns an action to each possible realization of $I^i_t$. The structure of agents' private signals and utilities remain the same. The speed of learning is likewise defined as before.

We now explain why, in this extended model, Theorem~\ref{main theorem} still holds as stated. The proof of Lemma~\ref{lemma:strategic-myopic} applies verbatim, as only one agent is considered, and the observation structure plays no role. Lemma~\ref{lem:rate-bound-from-observing} likewise still applies, but an adjustment needs to be made: The imitation principle \eqref{eq:imitation} again compares the loss the agent suffers in period $t$, assuming he suffers no loss in the future (on the left-hand-side), to the loss from always taking the action agent $j$ took the last time $t-D$ he was observed by agent $i$ (on the right-hand-side)
\begin{align*}
(1-\delta)(\bar{u}-\E[u(a_t^i,\omega)]) \leq \bar{u}-\E\left[u(a_{\tau(t)}^j, \omega)\right] \leq \max_{t' \in \{t-D-1,\ldots,t-1\}}\bar{u}-\E\left[u(a_{t'}^j, \omega)\right],
\end{align*}
where $\tau(t) =\max O_{i,j} \cap \{0,\ldots,t-1\}$ is the last time (before $t$) that $i$ observed $j$.
Now continuing as in the proof of Lemma~\ref{lem:rate-bound-from-observing} 
\begin{align*}
   \liminf_{t \to \infty} -\frac{1}{t} \log (\bar{u}-\E[u(a_t^i,\omega)])
    &\geq \liminf_{t \to \infty} -\frac{1}{t} \log (\bar{u}-\E[u(a_{t}^j,\omega)])
\end{align*}
which implies the result of  Lemma~\ref{lem:rate-bound-from-observing}. Here we crucially use the fact that there are at most $D$ periods between observations.

The proof of Lemma~\ref{prop1} remains valid, since the observation structure plays no role. The same holds for Lemma~\ref{lemma:private_signal}: The same proof applies to the modified version of the observed history $H^i_t$. Finally, the proof of Theorem~\ref{main theorem} again applies verbatim.

\bibliographystyle{ecta}
\bibliography{reference.bib}

\begin{thebibliography}{37}
\newcommand{\enquote}[1]{``#1''}
\expandafter\ifx\csname natexlab\endcsname\relax\def\natexlab#1{#1}\fi

\bibitem[\protect\citeauthoryear{Acemoglu, Dahleh, Lobel, and
  Ozdaglar}{Acemoglu et~al.}{2011}]{acemoglu2011bayesian}
\textsc{Acemoglu, D., M.~A. Dahleh, I.~Lobel, and A.~Ozdaglar} (2011):
  \enquote{Bayesian learning in social networks,} \emph{The Review of Economic
  Studies}, 78, 1201--1236.

\bibitem[\protect\citeauthoryear{Arieli, Babichenko, M{\"u}ller, Pourbabaee,
  and Tamuz}{Arieli et~al.}{2023}]{arieli2023hazards}
\textsc{Arieli, I., Y.~Babichenko, S.~M{\"u}ller, F.~Pourbabaee, and O.~Tamuz}
  (2023): \enquote{The Hazards and Benefits of Condescension in Social
  Learning,} \emph{arXiv preprint arXiv:2301.11237}.

\bibitem[\protect\citeauthoryear{Bala and Goyal}{Bala and
  Goyal}{1998}]{bala1998learning}
\textsc{Bala, V. and S.~Goyal} (1998): \enquote{Learning from neighbours,}
  \emph{The Review of Economic Studies}, 65, 595--621.

\bibitem[\protect\citeauthoryear{Banerjee}{Banerjee}{1992}]{banerjee1992simple}
\textsc{Banerjee, A.~V.} (1992): \enquote{A simple model of herd behavior,}
  \emph{The Quarterly Journal of Economics}, 107, 797--817.

\bibitem[\protect\citeauthoryear{Bikhchandani, Hirshleifer, and
  Welch}{Bikhchandani et~al.}{1992}]{BichHirshWelch:92}
\textsc{Bikhchandani, S., D.~Hirshleifer, and I.~Welch} (1992): \enquote{A
  theory of fads, fashion, custom, and cultural change as informational
  cascades,} \emph{Journal of Political Economy}, 992--1026.

\bibitem[\protect\citeauthoryear{Bolton and Harris}{Bolton and
  Harris}{1999}]{bolton1999strategic}
\textsc{Bolton, P. and C.~Harris} (1999): \enquote{Strategic experimentation,}
  \emph{Econometrica}, 67, 349--374.

\bibitem[\protect\citeauthoryear{Cover and Thomas}{Cover and
  Thomas}{2006}]{cover2006elements}
\textsc{Cover, T.~M. and J.~A. Thomas} (2006): \emph{Elements of information
  theory}, John Wiley \& Sons.

\bibitem[\protect\citeauthoryear{Dasaratha, Golub, and Hak}{Dasaratha
  et~al.}{2022}]{dasaratha2020learning}
\textsc{Dasaratha, K., B.~Golub, and N.~Hak} (2022): \enquote{{Learning from
  Neighbours about a Changing State},} \emph{The Review of Economic Studies},
  90, 2326--2369.

\bibitem[\protect\citeauthoryear{Dasaratha and He}{Dasaratha and
  He}{2019}]{dasaratha2019aggregative}
\textsc{Dasaratha, K. and K.~He} (2019): \enquote{Aggregative Efficiency of
  Bayesian Learning in Networks,} \emph{arXiv preprint arXiv:1911.10116}.

\bibitem[\protect\citeauthoryear{DeGroot}{DeGroot}{1974}]{degroot1974reaching}
\textsc{DeGroot, M.~H.} (1974): \enquote{Reaching a consensus,} \emph{Journal
  of the American Statistical Association}, 69, 118--121.

\bibitem[\protect\citeauthoryear{Dembo and Zeituni}{Dembo and
  Zeituni}{2009}]{dembo2009large}
\textsc{Dembo, A. and O.~Zeituni} (2009): \emph{Large deviations techniques and
  applications}, Springer, second ed.

\bibitem[\protect\citeauthoryear{Frick, Iijima, and Ishii}{Frick
  et~al.}{2023}]{frick2022learning}
\textsc{Frick, M., R.~Iijima, and Y.~Ishii} (2023): \enquote{Learning
  Efficiency of Multi-Agent Information Structures,} \emph{Journal of Political
  Economy, forthcming}.

\bibitem[\protect\citeauthoryear{Frongillo, Schoenebeck, and Tamuz}{Frongillo
  et~al.}{2011}]{frongillo2011social}
\textsc{Frongillo, R.~M., G.~Schoenebeck, and O.~Tamuz} (2011): \enquote{Social
  learning in a changing world,} in \emph{International Workshop on Internet
  and Network Economics}, Springer, 146--157.

\bibitem[\protect\citeauthoryear{Gale and Kariv}{Gale and
  Kariv}{2003}]{gale2003bayesian}
\textsc{Gale, D. and S.~Kariv} (2003): \enquote{Bayesian learning in social
  networks,} \emph{Games and economic behavior}, 45, 329--346.

\bibitem[\protect\citeauthoryear{Golub and Jackson}{Golub and
  Jackson}{2010}]{golub2010naive}
\textsc{Golub, B. and M.~O. Jackson} (2010): \enquote{Naive learning in social
  networks and the wisdom of crowds,} \emph{American Economic Journal:
  Microeconomics}, 2, 112--49.

\bibitem[\protect\citeauthoryear{Golub and Sadler}{Golub and
  Sadler}{2017}]{golub2017learning}
\textsc{Golub, B. and E.~Sadler} (2017): \enquote{Learning in social networks,}
  \emph{Available at SSRN 2919146}.

\bibitem[\protect\citeauthoryear{Graham and Pike}{Graham and
  Pike}{2008}]{graham2008note}
\textsc{Graham, A.~J. and D.~A. Pike} (2008): \enquote{A note on thresholds and
  connectivity in random directed graphs,} \emph{Atl. Electron. J. Math}, 3,
  1--5.

\bibitem[\protect\citeauthoryear{Grossman and Stiglitz}{Grossman and
  Stiglitz}{1980}]{grossman1980impossibility}
\textsc{Grossman, S.~J. and J.~E. Stiglitz} (1980): \enquote{On the
  impossibility of informationally efficient markets,} \emph{The American
  Economic Review}, 70, 393--408.

\bibitem[\protect\citeauthoryear{Hann-Caruthers, Martynov, and
  Tamuz}{Hann-Caruthers et~al.}{2018}]{hann2018speed}
\textsc{Hann-Caruthers, W., V.~V. Martynov, and O.~Tamuz} (2018): \enquote{The
  speed of sequential asymptotic learning,} \emph{Journal of Economic Theory},
  173, 383--409.

\bibitem[\protect\citeauthoryear{Harel, Mossel, Strack, and Tamuz}{Harel
  et~al.}{2021}]{harel2021rational}
\textsc{Harel, M., E.~Mossel, P.~Strack, and O.~Tamuz} (2021):
  \enquote{Rational groupthink,} \emph{The Quarterly Journal of Economics},
  136, 621--668.

\bibitem[\protect\citeauthoryear{Heidhues, Rady, and Strack}{Heidhues
  et~al.}{2015}]{heidhues2015strategic}
\textsc{Heidhues, P., S.~Rady, and P.~Strack} (2015): \enquote{Strategic
  experimentation with private payoffs,} \emph{Journal of Economic Theory},
  159, 531--551.

\bibitem[\protect\citeauthoryear{Kanoria and Tamuz}{Kanoria and
  Tamuz}{2013}]{kanoria2013tractable}
\textsc{Kanoria, Y. and O.~Tamuz} (2013): \enquote{Tractable Bayesian social
  learning on trees,} \emph{IEEE Journal on Selected Areas in Communications},
  31, 756--765.

\bibitem[\protect\citeauthoryear{Keller and Rady}{Keller and
  Rady}{2010}]{keller2010strategic}
\textsc{Keller, G. and S.~Rady} (2010): \enquote{Strategic experimentation with
  Poisson bandits,} \emph{Theoretical Economics}, 5, 275--311.

\bibitem[\protect\citeauthoryear{Keller, Rady, and Cripps}{Keller
  et~al.}{2005}]{keller2005strategic}
\textsc{Keller, G., S.~Rady, and M.~Cripps} (2005): \enquote{Strategic
  experimentation with exponential bandits,} \emph{Econometrica}, 73, 39--68.

\bibitem[\protect\citeauthoryear{Leskovec and Horvitz}{Leskovec and
  Horvitz}{2008}]{leskovec2008planetary}
\textsc{Leskovec, J. and E.~Horvitz} (2008): \enquote{Planetary-scale views on
  a large instant-messaging network,} in \emph{Proceedings of the 17th
  international conference on World Wide Web}, 915--924.

\bibitem[\protect\citeauthoryear{Migrow}{Migrow}{2022}]{migrow2022strategic}
\textsc{Migrow, D.} (2022): \enquote{Strategic Observational Learning,}
  \emph{arXiv preprint arXiv:2212.09889}.

\bibitem[\protect\citeauthoryear{Molavi, Tahbaz-Salehi, and Jadbabaie}{Molavi
  et~al.}{2018}]{molavi2018theory}
\textsc{Molavi, P., A.~Tahbaz-Salehi, and A.~Jadbabaie} (2018): \enquote{A
  theory of non-Bayesian social learning,} \emph{Econometrica}, 86, 445--490.

\bibitem[\protect\citeauthoryear{Moscarini, Ottaviani, and Smith}{Moscarini
  et~al.}{1998}]{moscarini1998social}
\textsc{Moscarini, G., M.~Ottaviani, and L.~Smith} (1998): \enquote{Social
  learning in a changing world,} \emph{Economic Theory}, 11, 657--665.

\bibitem[\protect\citeauthoryear{Moscarini and Smith}{Moscarini and
  Smith}{2002}]{moscarini2002law}
\textsc{Moscarini, G. and L.~Smith} (2002): \enquote{The law of large demand
  for information,} \emph{Econometrica}, 70, 2351--2366.

\bibitem[\protect\citeauthoryear{Mossel, Sly, and Tamuz}{Mossel
  et~al.}{2014}]{mossel2014asymptotic}
\textsc{Mossel, E., A.~Sly, and O.~Tamuz} (2014): \enquote{Asymptotic learning
  on bayesian social networks,} \emph{Probability Theory and Related Fields},
  158, 127--157.

\bibitem[\protect\citeauthoryear{Mossel, Sly, and Tamuz}{Mossel
  et~al.}{2015}]{mossel2015strategic}
---\hspace{-.1pt}---\hspace{-.1pt}--- (2015): \enquote{Strategic learning and
  the topology of social networks,} \emph{Econometrica}, 83, 1755--1794.

\bibitem[\protect\citeauthoryear{Rosenberg, Solan, and Vieille}{Rosenberg
  et~al.}{2009}]{rosenberg2009informational}
\textsc{Rosenberg, D., E.~Solan, and N.~Vieille} (2009): \enquote{Informational
  externalities and emergence of consensus,} \emph{Games and Economic
  Behavior}, 66, 979--994.

\bibitem[\protect\citeauthoryear{Rosenberg and Vieille}{Rosenberg and
  Vieille}{2019}]{rosenberg2019efficiency}
\textsc{Rosenberg, D. and N.~Vieille} (2019): \enquote{On the efficiency of
  social learning,} \emph{Econometrica}, 87, 2141--2168.

\bibitem[\protect\citeauthoryear{Smith and S{\o}rensen}{Smith and
  S{\o}rensen}{2000}]{smith2000pathological}
\textsc{Smith, L. and P.~S{\o}rensen} (2000): \enquote{Pathological outcomes of
  observational learning,} \emph{Econometrica}, 68, 371--398.

\bibitem[\protect\citeauthoryear{S{\o}rensen}{S{\o}rensen}{1996}]{sorensen1996rational}
\textsc{S{\o}rensen, P.~N.} (1996): \enquote{Rational social learning,} Ph.D.
  thesis, Massachusetts Institute of Technology.

\bibitem[\protect\citeauthoryear{Vives}{Vives}{1993}]{vives1993fast}
\textsc{Vives, X.} (1993): \enquote{How fast do rational agents learn?}
  \emph{The Review of Economic Studies}, 60, 329--347.

\bibitem[\protect\citeauthoryear{Watts and Strogatz}{Watts and
  Strogatz}{1998}]{watts1998collective}
\textsc{Watts, D.~J. and S.~H. Strogatz} (1998): \enquote{Collective dynamics
  of ‘small-world’networks,} \emph{nature}, 393, 440--442.

\end{thebibliography}

\end{document}